\documentclass[prl,reprint,twocolumn,aps,english,superscriptaddress,floatfix,longbibliography]{revtex4-2}
\usepackage{graphics}
\usepackage{graphicx}
\usepackage{here}
\usepackage{amsthm}
\usepackage{arydshln}
\usepackage{setspace}
\usepackage{wrapfig}
\usepackage{array,amsmath, amssymb, latexsym, bm, mathtools, braket, multirow, url}
\usepackage{color}
\usepackage{cancel}
\usepackage{comment}
\usepackage{ascmac}
\usepackage{tikz}
\usepackage{paralist}
\usepackage{hyperref}

\usepackage[all]{xy}
\usepackage{xr}

\theoremstyle{plain}
\newtheorem{theorem}{Theorem}
\newtheorem{definition}[theorem]{Definition}
\newtheorem{lemma}[theorem]{Lemma}
\newtheorem{proposition}[theorem]{Proposition}

\theoremstyle{definition}
\newtheorem{remark}[theorem]{Remark}

\newcommand{\cC}{\mathcal{C}}

\newcommand{\cH}{\mathcal{H}}

\newcommand{\cL}{\mathcal{L}}
\newcommand{\cM}{\mathcal{M}}

\newcommand{\cO}{\mathcal{O}}

\newcommand{\cS}{\mathcal{S}}
\newcommand{\cT}{\mathcal{T}}

\newcommand{\cV}{\mathcal{V}}

\DeclareMathOperator{\Tr}{Tr}
\DeclareMathOperator{\id}{id}

\newcommand{\Her}[1]{\cL_{\mathrm{H}}(#1)}
\newcommand{\Psd}[1]{\cL_{\mathrm{H}}^+(#1)}

\begin{document}

\title{Detection of Beyond-Quantum Non-locality based on Standard Local Quantum Observables}
\author{Hayato Arai}
\email{m18003b@math.nagoya-u.ac.jp}
\affiliation{Graduate School of Mathematics, Nagoya University, Furo-cho, Chikusa-ku, Nagoya, 464-8602, Japan}

\author{Baichu Yu}
\email{yubc@sustech.edu.cn}
\affiliation{Shenzhen Institute for Quantum Science and Engineering, 
Southern University of Science and Technology, Nanshan District, Shenzhen, 518055, China}
\affiliation{International Quantum Academy (SIQA), Shenzhen 518048, China}

\author{Masahito Hayashi}
\email{hmasahito@cuhk.edu.cn}
\affiliation{School of Data Science, The Chinese University of Hong Kong, Shenzhen, Longgang District, Shenzhen, 518172, China}
\affiliation{International Quantum Academy (SIQA), Shenzhen 518048, China}
\affiliation{Graduate School of Mathematics, Nagoya University, Furo-cho, Chikusa-ku, Nagoya, 464-8602, Japan}

\begin{abstract}
Device independent detections of quantum non-locality like 
Bell-CHSH inequality are 
important methods
to detect quantum non-locality
because the whole protocol can be implemented 
by uncertified local observables.
However, this detection is not sufficient for the justification of standard quantum theory,
because there are theoretically many types of beyond-quantum non-local states in General Probabilistic Theories.
One important class is Entanglement Structures (ESs),
which contain beyond-quantum non-local states 
even though their local systems are completely equivalent to standard quantum systems.
This paper shows that any device independent detection cannot distinguish beyond-quantum non-local states from standard quantum states.
To overcome this problem,
this paper gives a device dependent detection based on local observables to distinguish any beyond-quantum non-local state from all standard quantum states.
Especially, we give a way to detect any beyond-quantum non-local state in two-qubit ESs
by observing only spin observables on local systems.
\end{abstract}

\maketitle

\textit{Introduction}---%
Bell's inequality \cite{Bell1964} (or CHSH inequality \cite{CHSH1970}) is one of the important ways to detect quantum non-locality in our physical systems.
Bell-CHSH inequality (hereinafter, CHSH inequality) consists of bipartite players and their local operations.
It is especially important that the protocol of CHSH inequality can be implemented by local observables.
In other words,
by implementing the protocol of CHSH inequality as a bipartite communication task,
we can experimentally detect quantum non-locality of our physical systems when Bell-CHSH inequality is violated.
Actually, the violation of CHSH inequality is confirmed in physical experiments \cite{FC1972,ADR1982,WJS1998,ZZHE1993,RKM2001,SUKZ2010}.
Moreover, CHSH inequality can be implemented without certification of measurement devices.
Such detection without certification of measurement devices is called \textit{device independent} (DI) detection \cite{Review2014,NPA2008,Scarani2012,GKW2018,SB2020,ABG2007,PAB2009,VV2019,BBB2010}.
These remarkable results played an important role in the early studies of quantum physics and quantum information theory
to ensure that our physical systems truly possess quantum non-locality.

However,
it is not sufficient for the strict verification of quantum theory to detect standard quantum non-locality
because
there are many other theoretical models with non-locality than quantum systems.
Such models can be described as General Probabilistic Theories (GPTs)
\cite{PR1994,Plavala2017,Plavala2021,Pawlowski2009,Short2010,Barnum2012,Barrett2007,CDP2010,CS2015,CS2016,Muller2013,BLSS2017,Barnum2019,AH2022,Janotta2014,KBBM2017,Matsumoto2018,Takagi2019,Takakura2022,Arai2019,ALP2019,YAH2020,ALP2021,MAB2022,Barnum.Steering:2013}.
GPT is a framework for general theoretical models with states and measurements, including classical and quantum theories.
\textit{PR-box} \cite{PR1994,Plavala2017,Plavala2021} is a typical example of non-local models with beyond-quantum non-locality.
In PR-box, the CHSH value attains four even though the bound in quantum theory is given as $2\sqrt{2}$, known as Tirelson's bound \cite{Cirelson1980}.
In other words,
the CHSH inequality can detect the beyond-quantum non-locality in PR-box.

However, in contrast to models that can be detected by CHSH inequality,
there are models that cannot be detected by CHSH inequality
even though their local subsystems are completely equivalent to standard quantum systems.
Such models are called \textit{Entanglement Structures} (ESs) with local quantum subsystems \cite{Janotta2014,Arai2019,ALP2019,YAH2020,ALP2021,Barnum.Steering:2013},
including not only the Standard Entanglement Structure (SES), i.e., the standard quantum model defined by the tensor product
but also many other models.
Some ESs have fewer non-local states than the SES \cite{Arai2019,YAH2020},
and also, some ESs has beyond-quantum non-local states,
i.e.,
non-local states that do not belong to the SES \cite{AH2022,ALP2019,ALP2021,Barnum.Steering:2013}.
In order to ensure that our physical systems obey truly standard quantum theory,
it is also necessary to verify whether beyond-quantum non-local states exist or not.
However, preceding studies \cite{Banik:2012,Stevens:2013,Barnum.Steering:2013} have revealed that all ESs satisfy Tirelson's bound,
i.e.,
CHSH inequality cannot distinguish the SES from any beyond-quantum non-local state in ESs.

Furthermore,
as we show in this paper,
not only CHSH inequality,
but also any DI detection cannot distinguish any beyond-quantum non-local state from the SES.
Although a similar statement was shown by the reference \cite{BBB2010}, this paper also shows a corresponding statement in our setting with GPTs
 as Theorem~\ref{theorem:DI}.
Therefore, 
this paper deals with a device dependent detection of an arbitrary beyond-quantum non-local state in ESs by an experimental protocol.
First,
we give a device dependent detection separating an arbitrary given beyond-quantum state from all standard quantum states
as an inequality defined by local observables (Theorem~\ref{theorem:criterion}).
Next, we give a bipartite protocol to implement the above detection.
Our protocol consists of local operations by bipartite players Alice and Bob and classical communication by them.
In the protocol,
Alice and Bob detect whether a target state is beyond-quantum or not.
If the target state is truly beyond-quantum,
Alice and Bob conclude that the target state is beyond-quantum with high probability.

Our criterion and protocol are implemented by a complicated sequence of local observables in general.
However, in the 2-qubits case,
i.e.,
in the $2\times2$-dimensional case,
we give a simple detection of a beyond-quantum non-local state by observing Pauli's spin observables in a specific order.
It is known that maximally entangled states are detected when Alice and Bob observe Pauli's spin observable $\sigma_x,\sigma_y,\sigma_z$ in the same order with the sequence of coefficients $(1,1,-1)$ \cite{HMT2006,ZH2019}.
In contrast to this result,
we clarify that the sequence of coefficients $(1,1,1)$ detects a beyond-quantum non-local state.
Moreover, like Bell's scenario,
any beyond-quantum non-local ``pure'' state can be detected by sequential local observables biased in the same way as  $\sigma_x,\sigma_y,\sigma_z$ (Theorem~\ref{theorem:A3}).
As a result,
we give a convenient detection for beyond-quantum non-local pure states like Bell's inequality in the 2-qubits case.

\textit{The setting of GPTs and entanglement structures}---%
A model of GPTs is defined as follows.
\begin{definition}[A Model of GPTs]\label{def:model}
	A model of GPTs is defined by a tuple $\bm{G}=(\cV,\langle\ ,\ \rangle,\cC,u)$,
	where $(\cV, \langle\ ,\ \rangle)$, $\cC$, and $u$ are a real-vector space with inner product, a proper cone, and an order unit of $\cC^\ast$, respectively.
\end{definition}
For a model of GPTs $\bm{G}$,
state space and measurement space are defined as follows.
\begin{definition}[State Space of GPTs]
	Given a model of GPTs $\bm{G}=(\cV,\langle\ ,\ \rangle,\cC,u)$,
	the state space of $\bm{G}$ is defined as
	\begin{align}\label{def:state}
		\cS(\bm{G}):=\left\{\rho\in\cV\middle|\langle\rho,u\rangle=1\right\}.
	\end{align}
	Here, we call an element $\rho\in\cS(\bm{G})$ a state of $\bm{G}$.
\end{definition}


\begin{definition}[Measurements of GPTs]\label{def:measurement}
	Given a model  of GPTs $\bm{G}=(\cV,\langle\ ,\ \rangle,\cC,u)$,
	we say that a family $\{M_i\}_{i\in I}$ is a measurement
	if $M_i\in\cC^\ast$ and $\sum_{i\in I} M_i=u$.
	Besides, the index $i$ is called an outcome of the measurement.
	Here, the set of all measurements is denoted as $\cM(\bm{G})$.
	Especially, the set of measurements with $n$-number of outcomes is denoted as
	$\cM_n(\bm{G})$.
\end{definition}

In this setting,
state space and measurement space are always convex.
We call an element \textit{pure} when the element is extremal.
Also, when a state $\rho\in\cS(\bm{G})$ is measured by a measurement $\{M_i\}\in\cM(\bm{G})$,
the probability $p_i$ to get an outcome $i$ is given as
\begin{align}
	p_i:=\langle \rho, M_i\rangle.
\end{align}

Quantum theory is a typical example of a model of GPTs,
i.e.,
quantum theory is given as a model $\bm{G}=(\cL_{\mathrm{H}}(\cH), \Tr,\cL^+_{\mathrm{H}}(\cH)),I)$,
where $\cL_{\mathrm{H}}(\cH)$ and $\cL^+_{\mathrm{H}}(\cH)$ be the set of Hermitian matrices and the set of positive semi-definite matrices on a finite dimensional Hilbert space $\cH$, respectively.
In this model,
the state space is the set of density matrices,
also the measurement space is the set of Positive Operator Valued Measures (POVMs).

In order to discuss Bell scenario,
we introduce the model of composite systems in GPTs.
Given two models in GPTs $\bm{G}_A=(\cV_A,\langle\ ,\ \rangle_A,\cC_A,u_A)$ and $\bm{G}_B=(\cV_B,\langle\ ,\ \rangle_B,\cC_B,u_B)$,
a model $\bm{G}=(\cV,\langle\ ,\ \rangle,\cC,u)$ of composite system of two models is defined as follows \cite{Plavala2021}:
The vector space $\cV$ is defined as the tensor product $\cV=\cV_A\otimes\cV_B$.
The inner product $\langle\ ,\ \rangle$ is defined as the induced inner product by the tensor product,
i.e.,
the inner product $\langle x_1,x_2 \rangle$ is defined as $\langle x_1,x_2 \rangle=\sum_{i,j}\langle a_1^{(i)},a_2^{(j)}\rangle_A \langle b_1^{(i)},b_2^{(j)}\rangle_B$ for elements $x_1=\sum_i a_1^{(i)}\otimes b_1^{(i)}\in\cV$ and $x_2=\sum_j a_2^{(j)}\otimes b_2^{(j)}\in\cV$.
The cone $\cC$ is chosen as a cone satisfying the inequality
\begin{align}
	\cC_A\otimes\cC_B\subset\cC\subset\left(\cC_A^\ast\otimes\cC_B^\ast\right)^\ast,\label{eq:composite}
\end{align}
where the tensor product $\cC_A\otimes\cC_B$ is defined as
\begin{align}\label{eq:tensor}
	\cC_A\otimes\cC_B:=\left\{\sum_{ij} a_i\otimes b_j \middle| a_i\in\cC_A, b_j\in\cC_B\right\}.
\end{align}
The order unit $u$ is defined as $u=u_A\otimes u_B\in\cV$.

Especially,
in this paper,
we mainly consider \textit{Entanglement Structures (ESs) with local quantum subsystems} (hereinafter, we simply call them ESs),
i.e.,
models of composite systems $\bm{G}$ whose local systems $\bm{G}_A$ and $\bm{G}_B$ are equivalent to quantum theory.
By modifying the above conditions of the definition of composite systems,
we define an ES as follows.
\begin{definition}[Entanglement Structure \cite{Plavala2021,Janotta2014,ALP2019,AH2022}]
	We say that a model $\bm{G}=(\cL_\mathrm{H}(\cH_A\otimes\cH_B), \Tr, \cC,I)$ is an entanglement structure
	if $\cC$ satisfies
	\begin{align}\label{eq:quantum}
	\mathrm{SEP}(A;B)\subset\cC\subset\mathrm{SEP}(A;B)^\ast,
\end{align}
where the proper cone $\mathrm{SEP}(A;B)$ is defined as
\begin{align}\label{eq:sep}
	\mathrm{SEP}(A;B):=\Psd{\cH_A}\otimes\Psd{\cH_B}.
\end{align}
\end{definition}
Since an ES is characterized by a positive cone with \eqref{eq:quantum},
it is identified with the corresponding positive cone $\cC$.

The inclusion relation \eqref{eq:quantum} implies that
a model of  quantum composite system is not uniquely determined.
On the other hand,
standard quantum systems obey the only model $\mathrm{SES}(A;B):=\Psd{\cH_A\otimes\cH_B}$.
In this paper, we call this model the \textit{Standard Entanglement Structure} (SES).
The set $\mathrm{SEP}(A;B)^\ast$ contains non-positive Hermitian matrices,
and therefore,
a non-positive state is available in an ES $\cC\not\subset\mathrm{SES}(A;B)$.
We call a non-positive state,
i.e.,
a state $\rho\in\cS(\mathrm{SEP}^\ast(A;B))\setminus\cS(\mathrm{SES}(A;B))$ a \textit{beyond-quantum state}.
Our interest is how we detect beyond-quantum states if they exist.

\textit{Impossibility of Device-Independent Detection for Beyond-Quantum States}---%
First, we consider the possibility of the \textit{Device-Independent} (DI) detection of 
a beyond-quantum state (Figure~\ref{figure_DI}).
In the device-independent detection,
we have no certificate of measurement devices.
Therefore, 
it is natural to consider that
a beyond-quantum state $\rho_0\in\cS(\bm{SEP}^\ast(A;B))$ is 
distinguished device-independently 
by local measurements $M_{a}^{A}:=\{M_{a;i}^{A}\}_{i\in I}$ and $M_{b}^{B}:=\{M_{b;j}^{B}\}_{j\in J}$ 
from all standard quantum states 
when 
no pair of 
a standard quantum state and
local POVMs simulates the pair of  
the state $\rho_0$ and 
local POVMs $M_{a}^{A}:=\{M_{a;i}^{A}\}_{i\in I}$ and $M_{b}^{B}:=\{M_{b;j}^{B}\}_{j\in J}$, i.e.,
there does not exist a pair of
a standard quantum state
$\rho_1\in\cS(\mathrm{SES}(A;B))$
and 
local POVMs $M_{a}^{A}:=\{M_{a;i}^{A\prime}\}_{i\in I}$ and $M_{b}^{B}:=\{M_{b;j}^{B\prime}\}_{j\in J}$
such that the relation
\begin{align}
\begin{aligned}
\Tr \rho_0 M_{a;i}^{A}\otimes M_{b;j}^{B}
= \Tr \rho_1 M_{a;i}^{A\prime}\otimes M_{b;j}^{B\prime}
	\end{aligned}\label{DI-detect}
\end{align}
holds for any $a,b,i,j$.
In other words,
a beyond-quantum state $\rho_0\in\cS(\bm{SEP}^\ast(A;B))$ is 
distinguished device-independently 
from all standard quantum states 
when 
there exist local measurements $M_{a}^{A}:=\{M_{a;i}^{A}\}_{i\in I}$ and $M_{b}^{B}:=\{M_{b;j}^{B}\}_{j\in J}$ 
to satisfy the above condition.
Therefore, the above device-independent detectability 
is equivalent to 
the impossibility of the simulation by a pair of 
a standard quantum state and local POVMs.

\begin{figure}[h]
	\centering
	\includegraphics[width=6cm]{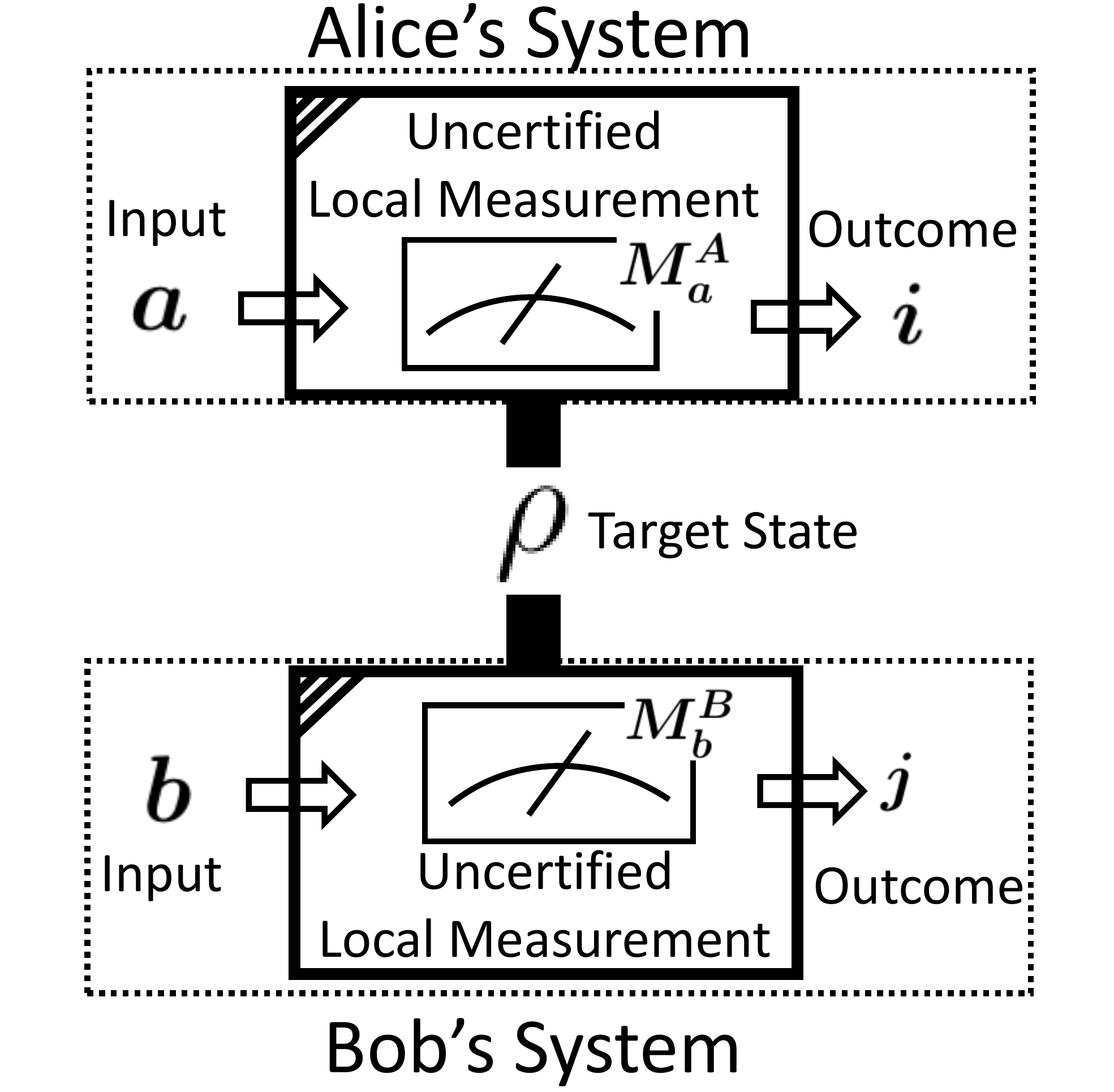}
	\caption{
	In DI detection, Alice and Bob apply uncertified local measurements $M_{a}^{A}:=\{M_{a;i}^{A}\}_{i\in I}$ and $M_{b}^{B}:=\{M_{b;j}^{B}\}_{j\in J}$ to a given non-local state $\rho$.
	Then, Alice and Bob determine whether $\rho$ is beyond-quantum by the probability $\Tr \rho M_{a;i}^{A}\otimes M_{b;j}^{B}$.
	}
	\label{figure_DI}
\end{figure}

However, previous studies \cite{Banik:2012,Stevens:2013,Barnum.Steering:2013} showed that CHSH inequality cannot detect any beyond-quantum states by noticing steering condition.
Furthermore,
the following theorem holds.

\begin{theorem}\label{theorem:DI}
For any pair of 
a beyond-quantum state $\rho_0\in\cS(\bm{SEP}^\ast(A;B))$ 
and 
local POVMs $M_{a}^{A}:=\{M_{a;i}^{A}\}_{i\in I}$ and $M_{b}^{B}:=\{M_{b;j}^{B}\}_{j\in J}$,
there exists
a pair of 
a standard quantum state
$\rho_1\in\cS(\mathrm{SES}(A;B))$
and 
local POVMs $M_{a}^{A}:=\{M_{a;i}^{A\prime}\}_{i\in I}$ and $M_{b}^{B}:=\{M_{b;j}^{B\prime}\}_{j\in J}$
to satisfy the condition \eqref{DI-detect}.
\end{theorem}

Although the reference \cite{BBB2010} proved a similar statement, it does not formulate the problem with GPTs.
Further,
while 
the proof in \cite{BBB2010} has a problem caused by
an inverse of a key operator,
our proof does not have such a problem
because our proof is straightforward and different from that of 
the reference \cite{BBB2010}, as shown in Appendix.
Due to Theorem~\ref{theorem:DI},
it is impossible to 
distinguish a beyond-quantum state 
from all standard quantum states.
To resolve this problem, 
instead of measurement devices without certification,
we need to employ 
measurement devices that are identified 
with certifications.
This problem setting is called device-dependent (DD) detection.

\textit{Device Dependent Detection of  Beyond-Quantum State and Its Implementation}---%

Now, we discuss a DD detection of an arbitrary given beyond-quantum state in ESs.
In the following analysis,
instead of the joint distribution,
as a simple indicator,
we focus on the sum of an expectation of a function 
$f(a,i,b,j)$, i.e.,
$\sum_{a,i,b,j} f(a,i,b,j)\Tr \rho M_{a;i}^{A}\otimes M_{b;j}^{B}$
so that
the magnitude relationship of this indicator
makes the required discrimination. 
For our simple analysis, we assume $f(a,i,b,j)=f(a,i)f(b,j)$.
Then, this value 
can be rewritten as

\begin{align}
\begin{aligned}\label{eq:pro-mea2}
	&\sum_{a,i,b,j} f(a,i,b,j)\Tr \rho M_{a;i}^{A}\otimes M_{b;j}^{B}\\
	=&\sum_{a,b} \Tr \rho \cO_a^A\otimes \cO_b^B,
\end{aligned}
\end{align}
where $\cO_a^{A}:=\sum_i f(a,i)M_{a;i}^{A}, \ \cO_b^{B}:=\sum_j f(b,j)M_{b;j}^{B}$.
The Hermitian matrices $\cO_a^{A}$ and $\cO_b^{B}$ can be regarded as standard quantum observables with the POVMs $M_{a;i}^{A}, \ M_{b;j}^{B}$ and outcomes $f(a,i), \ f(b,j)$, respectively.
Therefore, the value $\Tr \rho \cO_a^A\otimes \cO_b^B$ corresponds to the expectation value of the standard quantum observable 
$\cO_a^A\otimes \cO_b^B$ with the state $\rho$.
Hereinafter, we abbreviate the pair of POVMs and outcomes
in the left-hand side of \eqref{eq:pro-mea2} to the right-hand side of \eqref{eq:pro-mea2}
by using observables, according to this correspondence.

Based on the sum of the expectation of
standard quantum local observables,
the following theorem gives 
a DD detection of any beyond-quantum state from all standard quantum states.

\begin{theorem}\label{theorem:criterion}
	Given an arbitrary state $\rho_0\in\cS(\mathrm{SEP}^\ast(A;B))\setminus\cS(\mathrm{SES}(A;B))$,
	there exist families of local observables $\{\cO_{k}^{A}\}_{k=1}^{m}$ 
	and $\{\cO_{k}^{B}\}_{k=1}^{m}$ and a real number $\alpha$ satisfying the following two properties:
	\begin{enumerate}
		\item $\Tr \rho_0 \sum_{k=1}^{m} \cO_{k}^{A}\otimes\cO_{k}^{B} >\alpha$.
		\item $\displaystyle \sup_{\rho_1\in\cS(\mathrm{SES}(A;B))}\Tr \rho_1 \sum_{k=1}^m \cO_{k}^{A}\otimes\cO_{k}^{B} \le\alpha$.
	\end{enumerate}
\end{theorem}
The proof of Theorem~\ref{theorem:criterion} is written in Appendix,
but we remark that we can find $\{\cO_{k}^{A}\otimes\cO_{k}^{B}\}_{k=1}^m$ and $\alpha$ by a deterministic way.
Theorem~\ref{theorem:criterion} guarantees that 
the joint distribution with $\rho_0$
cannot be simulated by 
the joint distribution with any standard quantum state
$\rho_1$ under the common local measurements.
The above discussion can be understand in terms of 
the Semi-Definite Programing (SDP)
with the target function
$\Tr \rho_1 \sum_{k=1}^m \cO_{k}^{A}\otimes\cO_{k}^{B}$
and the trace $1$ condition.
The second relation in Theorem~\ref{theorem:criterion} shows that the solution of the SDP is upper bounded by $\alpha$.
The first relation in Theorem~\ref{theorem:criterion} 
states that 
$\rho_0$ attains a strictly larger value than the solution,
and therefore, $\rho_0$ is not positive semi-definite, i.e., beyond-quantum.

Next, we see that the detection given by Theorem~\ref{theorem:criterion} is implemented as the following DD detection protocol on bipartite scenario (Figure~\ref{fig:protocol}).

\begin{itemize}
\item \textbf{Aim and Strategy}
\vspace{-1ex}
\begin{itemize}
	\item Alice and Bob aim to determine whether a given target global state $\rho$ is beyond-quantum or not.
	\item Alice and Bob choose  $\{\cO_{k}^{A}\otimes\cO_{k}^{B}\}_{k=1}^m$ and $\alpha$ given in Theorem~\ref{theorem:criterion} based on their prediction that the target state $\rho$ is close to a beyond-quantum state $\rho_0$.
	\item Alice and Bob repeat the following protocol by $nm$-times for sufficiently large $n$.
\end{itemize}
\item \textbf{The Whole Protocol}
\vspace{-1ex}
\begin{enumerate}
	\item \textbf{Set Up}: Alice and Bob prepare a generator of the target state $\rho$.
	The generator always transmits the same global state $\rho$.
	\item \textbf{$l$-th Round}:
	The generator transmits the state $\rho$ to the composite system of Alice's and Bob's systems.
	Alice and Bob measure their local observables $\cO^{\rm{A}}_{k}$ and $\cO^{\rm{B}}_{k}$ with $l=qn+k$ ($1\le k \le m$, $q$ is the integer part of the quotient $l/n$).
	As a result, they get outcomes $o^{\rm{A}}_{l}$ and $o^{\rm{B}}_{l}$, respectively.
	\item \textbf{Determination}: Alice and Bob share their outcomes with classical communication. They then calculate the value $\frac{1}{n}\sum_{l=1}^{nm} o^{\rm{A}}_{l} o^{\rm{B}}_{l}$.
	If the inequality $\frac{1}{n}\sum_{l=1}^{nm} o^{\rm{A}}_{l} o^{\rm{B}}_{l}>\alpha$ holds,
	Alice and Bob conclude that the target state $\rho$ is beyond-quantum.
\end{enumerate}
\item \textbf{Justification}
\vspace{-1ex}
\begin{itemize}
	\item  On the limit $n\to\infty$, the value $\frac{1}{n}\sum_{l=1}^{nm} o^{\rm{A}}_{l} o^{\rm{B}}_{l}$ approximates the expectation value $\sum_{k=1}^m \Tr \rho \cO_{k}^{A}\otimes\cO_{k}^{B}$.
	\item If $n$ is sufficiently larger and the inequality $\frac{1}{n}\sum_{l=1}^{nm} o^{\rm{A}}_{l} o^{\rm{B}}_{l}>\alpha$ holds, Theorem~\ref{theorem:criterion} ensures that the target state $\rho$ is beyond-quantum with sufficiently large probability.
	\item If $\frac{1}{n}\sum_{l=1}^{nm} o^{\rm{A}}_{l} o^{\rm{B}}_{l}\le\alpha$ holds for sufficiently large $n$, their prediction $\rho_0$ is sufficiently different from the target state $\rho$.
\end{itemize}
\end{itemize}

\begin{figure}[h]
	\centering
	\includegraphics[width=6cm]{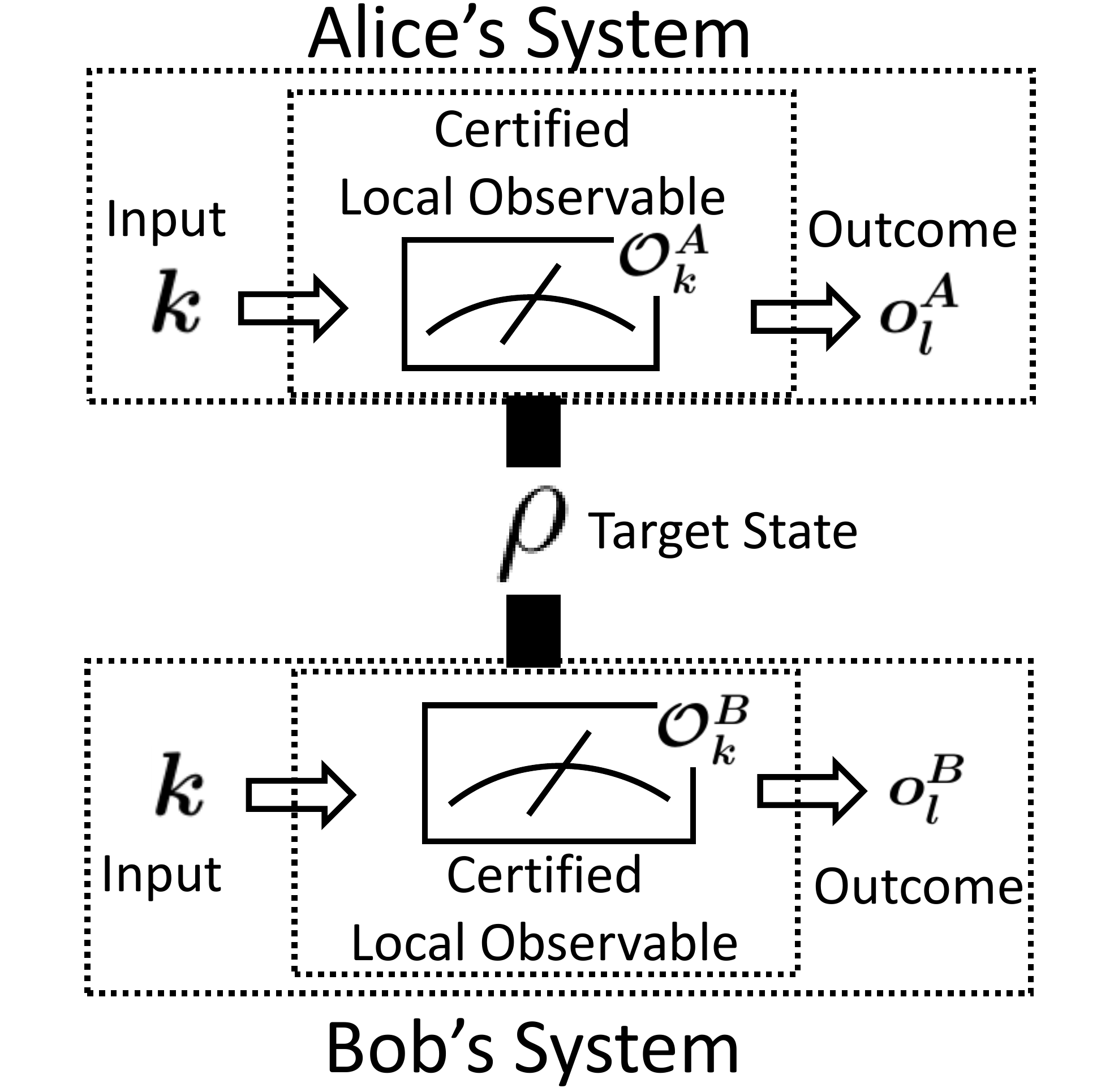}
	\caption{
	The $l$-th Round of the Detection Protocol of Criterion Given in Theorem~\ref{theorem:criterion}.
	Alice and Bob aim to detect whether $\rho$ is beyond-quantum.
	Alice and Bob prepare only their certified local observables with a certain order given in Theorem~\ref{theorem:criterion},
	and they estimate the expectation value in Theorem~\ref{theorem:criterion} as the average of all outcomes gotten in $nm$-rounds of observation.
	}
	\label{fig:protocol}
\end{figure}

In this way,
any beyond-quantum state can be detected by a finite number of certified local quantum observables with large probability.
In general cases,
this detection require large costs because we need to certify a number of local quantum observables dependent on a target state.
However,
in the $2\times 2$ dimensional case,
we give a detection of any beyond-quantum pure state with the certification of only three observables of Pauli's spin observables.

Let us consider an arbitrary ES $\bm{G}=(\cT(\cH),\Tr,\cC,I)$ of two local quantum systems with dimension 2,
i.e.,
we consider the case $\dim(\cH_A)=\dim(\cH_B)=2$.
First, we define the following function $\mathbb{A}_{\mathrm{Pauli}}(\rho;U_A,U_B)$ using Pauli's spin matrices:
\begin{align}
	&\mathbb{A}_{\mathrm{Pauli}}(\rho;U_A,U_B)
	\nonumber \\
	:=&\sum_{c=x,y,z} \Tr \left(U_A\otimes U_B\right)\rho\left(U_A^\dag\otimes U_B^\dag\right) \sigma_c\otimes\sigma_c,
\end{align}
where $U_A,U_B$ and $\sigma_x,\sigma_y,\sigma_z$ denote unitary matrices on $\cH_A,\cH_B$, 
Pauli's spin observables defined as
\begin{align}\label{eq:ob}
	\sigma_x:=
	\left(
	\begin{array}{cc}
		0&1\\
		1&0
	\end{array}
	\right),\quad
	\sigma_y:=
	\left(
	\begin{array}{cc}
		0&-i\\
		i&0
	\end{array}
	\right),\quad
	\sigma_z:=
	\left(
	\begin{array}{cc}
		1&0\\
		0&-1
	\end{array}
	\right),
\end{align}
respectively.

The value $\mathbb{A}_{\mathrm{Pauli}}(\rho;U_A,U_B)$ satisfies the following properties.
\begin{theorem}\label{theorem:A3}
	The following two properties hold:
	\begin{enumerate}
		\item For any beyond-quantum pure state $\rho_0\in\cS(\mathrm{SEP}^\ast(A;B))$, there exist unitary matrices $U_A,U_B$ such that $\mathbb{A}_{\mathrm{Pauli}}(\rho_0;U_A,U_B)>1$.
		\item $\displaystyle \sup_{\substack{U_A,U_B \ : \ \mathrm{\mbox{unitary}}\\\rho_1\in\cS(\bm{SES}(A;B))}}\mathbb{A}_{\mathrm{Pauli}}(\rho_1;U_A,U_B)\le1$.
	\end{enumerate}
\end{theorem}

The proof of Theorem~\ref{theorem:A3} is written in Appendix.
Theorem~\ref{theorem:A3} implies that
in $2\times 2$ dimensional case,
if a target state $\rho$ is beyond-quantum pure, there exists a pair of unitary matrices $U_A$ and $U_B$ such that
$\mathbb{A}_{\mathrm{Pauli}}(\rho;U_A,U_B)$ detects the beyond-quantum pure state $\rho_0$ from all standard quantum states $\rho_1$.
If we can apply the unitary operations in the whole protocol,
it is not necessary to certify the description of unitary matrices.
In other words, we only need to certify the observables $\sigma_x,\sigma_y,\sigma_z$ for all detections of any beyond-quantum pure state.
The criterion given in Theorem~\ref{theorem:A3} is approximately implemented by a reiteration of the protocol in Figure~\ref{fig:protocol} without the certification of unitary operations.

\textit{Conclusion}---%
In this paper,
we have discussed the detection of a beyond-quantum state in ESs of GPTs.
Even though local systems of ESs are equivalent to standard quantum systems,
we have shown that any device independent detection, including CHSH inequality, cannot separate any beyond-quantum state in ESs from the standard quantum states in the SES.
In contrast to device independent detection,
we have given a device dependent detection of an arbitrary given beyond-quantum state in any ES from all states in the SES based on local standard quantum observables.
Also, we have given an experimental implementation of the detection as a bipartite protocol.
The detection needs a large number of observables in general.
However, in the 2-qubits case,
we have given a simple detection based on Pauli's spin observables.
The detection can be implemented only by certified Pauli's spin observables and uncertified unitary operations.

An interesting remaining study is the development to a method to verify that no beyond-quantum state exists. In this case, the method need to be independent of a beyond-quantum state.
Another remaining study is a strict estimation of the error probability of the detection protocol.
In order to discuss how surely our models contains beyond-quantum non-local states,
we need to discuss the protocol in the context of hypothesis testing.

\textit{Acknowledgments}---%
HA is
supported by a JSPS Grant-in-Aids for JSPS Research Fellows No. JP22J14947.
MH is supported in part by the National Natural Science Foundation of China (Grant No. 62171212).

\section{Appendix}

\subsection{Proof of Theorem \ref{theorem:DI}}
We choose an arbitrary beyond-quantum state
 $\rho\in\cS(\bm{SEP}^\ast(A;B))$,
 and define the marginal state $\rho^{A}:=\Tr_B \rho$.
We define the subspace $\cH_A'$ of $\cH_A$
to be the range of $\rho^{A}$.
We denote the projection to $\cH_A'$ by $P$.

Then, we choose a basis of $\cH_A'$ and denote its dimension
by $d_A$.
Based on this basis, 
we denote the matrix on $\cH_A'$ which values 1 only on the $(i,j)$-th entry as $e_{ij}$,
and we define a maximally entangled state $\Phi$ 
on $\cH_A'\otimes \cH_A'$ as
\begin{align}
	\Phi:=\frac{1}{d_A}\sum_{ij}e_{ij}\otimes e_{ij}.
\end{align}
We consider the state
\begin{align}
	\rho':=\left(\frac{1}{\sqrt{d_A}}\sqrt{\rho^{A}}^{-1}\otimes I\right)\rho\left(\frac{1}{\sqrt{d_A}}\sqrt{\rho^{A}}^{-1}\otimes I\right)
\end{align}
on $\cH_A'\otimes \cH_B$.

The matrix $\rho'$ also belongs to $\cS(\bm{SEP}^\ast(A;B))$.
Therefore, Choi-Jamio\l kowski isomorphism \cite{Jami1972} ensures the existence of positive map $F:\Her{\cH_A'}\to\Her{\cH_B}$ satisfying $\rho'=(\id\otimes F)(\Phi)$.
Besides, the two equations $\Tr_B \rho'=\frac{1}{d_A}I$ and
\begin{align}
	\Tr_B (\id\otimes F)(\Phi)
	=\frac{1}{d_A}\sum_{ij} e_{ij}\otimes \Tr F(e_{ij})
\end{align}
imply that the map $F$ is trace preserving.
Therefore,
the state $\rho$ is rewritten as
\begin{align}
	&\rho
	=(\sqrt{d_A \rho^{A}}\otimes I)
	\rho'(\sqrt{d_A \rho^{A}}\otimes I)\nonumber \\
	=&(\sqrt{d_A \rho^{A}}\otimes I)
	\big((\id\otimes F)(\Phi)\big)(\sqrt{d_A \rho^{A}}\otimes I)
	\nonumber \\
	=&
	(\id\otimes F)(\sigma),
\end{align}
where $\sigma:=(\sqrt{d_A\rho^{A}}\otimes I)
\Phi(\sqrt{d_A \rho^{A}}\otimes I)$ 
is a standard quantum entangled state 
on $\cH_A'\otimes \cH_B$.

Next, any local POVMs $\{M_{a;i}^{A}\},\ 
\{M_{b;j}^{B}\}$ satisfy the following equation:
\begin{align}
	&\Tr \rho \left(M_{a;i}^{A}\otimes M_{b;j}^{B}\right)
	=\Tr \rho \left(P M_{a;i}^{A} P\otimes M_{b;j}^{B}\right)\nonumber \\
	&=\Tr \left(\id\otimes F\right)(\sigma)
	\left(P M_{a;i}^{A} P\otimes M_{b;j}^{B}\right)\nonumber \\
	&=\Tr \sigma\left(P M_{a;i}^{A} P \otimes F^\dag\left(M_{b;j}^{B}\right)\right)\nonumber \\
	&=\Tr \sigma\left( M_{a;i}^{A}  \otimes F^\dag\left(M_{b;j}^{B}\right)\right)
\label{eq:pro-mea}
\end{align}
because $P\sigma P=\sigma$.
Since the map $F$ is positive and trace preserving,
the adjoint map $F^\dag$ is positive and unital.
Therefore, the family 
$\left\{M_{b;j}^{B\prime}:=F^\dag\left(M_{b;j}^{B}\right)\right\}_j$ 
is a measurement in the system $B$.
As a result, we obtain the following equation
\begin{align}
	\Tr \rho \left(M_{a;i}^{A}\otimes M_{b;j}^{B}\right)
	=\Tr \sigma \left(M_{a;i}^{A}\otimes M_{b;j}^{B\prime}\right)\label{eq:pro-mea}
\end{align}
for any $a,b,i,j$.

\begin{remark}
The reference \cite{BBB2010}
also consider the inverse of a certain map $M$.
However, this map is not invertible in general.
Hence, the reference \cite{BBB2010} consider an invertible 
map $M_\epsilon$ with a parameter $\epsilon$
such that $\lim_{\epsilon \to 0}M_\epsilon$ equals the non-invertible map $M$.
In this case, 
a pair of 
a standard quantum state
$\rho_1\in\cS(\mathrm{SES}(A;B))$
and 
local POVMs $M_{a}^{A}:=\{M_{a;i}^{A\prime}\}_{i\in I}$ and $M_{b}^{B}:=\{M_{b;j}^{B\prime}\}_{j\in J}$
depends on $\epsilon$.
The reference \cite{BBB2010}
did not discuss the convergence of this pair under the limit 
$\epsilon \to 0$ while its inverse is employed.
However, since our method does not employ 
such an approximation,
our method works even with a non-invertible state $\rho^A$.
\end{remark}

\subsection{Proof of Theorem~\ref{theorem:criterion}}

As the proof of Theorem~\ref{theorem:criterion},
we give the following deterministic way to find the observables $\{\cO_k^{A}\otimes\cO_k^{B}\}_{k=1}^m$ and a real number $\alpha$ for any given state $\rho_0\in\cS(\mathrm{SEP}^\ast(A;B))\setminus\cS(\mathrm{SES}(A;B))$.

At first,
because the state space $\cS(\mathrm{SES}(A;B))$ is a closed convex set and the relation $\rho_0\not\in\cS(\mathrm{SES}(A;B))$,
hyperplane separation theorem \cite{BV2004} ensures the existence of a hyperplane that separates $\rho_0$ from the convex set $\cS(\mathrm{SES}(A;B))$.
In other words,
there exist an element $x\in\Her{\cH_A\otimes\cH_B}$ and a real number $\alpha\in\mathbb{R}$
such that 
$\Tr x\rho_0>\alpha$ and $\Tr x\rho_1\le\alpha$ for any $\rho_1\in\cS(\mathrm{SES}(A;B))$.
In practical situation,
we need to find such a Hermitian matrix $x$ by an analytical way,
but $x=\rho_0$ separates $\rho_0$ and $\cS(\mathrm{SES}(A;B))$ as follows.
Any $\rho_1\in\cS(\mathrm{SES}(A;B))$ satisfies the following inequality:
\begin{align}
	\left|\Tr \rho_0\rho_1\right|\stackrel{(a)}{<}\|\rho_0\|_2\|\rho_1\|_2\stackrel{(b)}{\le}\|\rho_0\|_2.
\end{align}
The inequality $(a)$ is shown by Schwarz inequality and its equality condition.
The equality condition of Schwarz inequality holds only when $\rho_1$ is proportional to $\rho_0$,
which never holds because $\rho_1$ is positive semi-definite and $\rho_0$ is not positive semi-definite.
The inequality $(b)$ is shown by $\|\rho_1\|\le 1$ for any $\rho_1\in\cS(\mathrm{SES}(A;B))$.
On the other hand,
the equation $\Tr \rho_0^2=\|\rho_0\|_2$ holds by definition,
therefore we find at least one element $x=\rho_0$ separating $\rho_0$ and $\cS(\mathrm{SES}(A;B))$.
Due to the latter discussion,
we consider general separation $x$ here.

Next, we formulate the element $x$ as a tensor product form.
Because the element $x$ belongs to the vector space $\Her{\cH_A\otimes\cH_B}=\Her{\cH_A}\otimes\Her{\cH_B}$,
the element $x$ is written as $x=\sum_{k=1}^m x^{A}_k \otimes x^{B}_k$, where $x^{A}_k\in\Her{\cH_A}$ and $x^{B}_k\in\Her{\cH_B}$.
As seen in the latter discussion,
the Hermitian matrices $x^{A}_k$ and $x^{B}_k$ can be regarded as observables $\cO_k^{A}$ and $\cO_k^{B}$, respectively.
Finally,
the observables $\cO_k^{A}$ and $\cO_k^{B}$ satisfy the equation
\begin{align}
	\Tr \rho_0 \sum_{k=1}^m \cO_k^{A}\otimes\cO_k^{B}=\Tr x\rho
\end{align}
for any $\rho\in\cS(\mathrm{SEP}^\ast(A;B))$.

\subsection{Proof of Theorem~\ref{theorem:A3}}

For the proof of Theorem~\ref{theorem:A3},
First,
we introduce the following function $\mathbb{A}'_{\mathrm{Pauli}}(\rho)$ as
\begin{align}
	\mathbb{A}'_{\mathrm{Pauli}}(\rho):=&\sum_{i=x,y,z} \left\langle \sigma_i\otimes\sigma_i\right\rangle_\rho
	=
	\Tr
	\left(
	\begin{array}{cccc}
		1&0&0&0\\
		0&-1&2&0\\
		0&2&-1&0\\
		0&0&0&1
	\end{array}
	\right)\rho.\label{eq:state-A}
\end{align}
Then,
we show the following lemmas.

\begin{lemma}\label{theorem:A1}
	The function $\mathbb{A}'_{\mathrm{Pauli}}(\rho)$ satisfies the following two properties:
	\begin{enumerate}
		\item $\mathbb{A}'_{\mathrm{Pauli}}(\rho)\le 1$ for any $\rho\in\cS(\bm{SES}(A;B))$.
		\item $\mathbb{A}'_{\mathrm{Pauli}}(\rho_{\mathrm{max}})=3$, where $\rho_{\mathrm{max}}\in\cS(\bm{SEP}^\ast(A;B))$ is defined as
		\begin{align}
			\rho_{\mathrm{max}}
			=
			\left(
			\begin{array}{cccc}
				\frac{1}{2}&0&0&0\\
				0&0&\frac{1}{2}&0\\
				0&\frac{1}{2}&0&0\\
				0&0&0&\frac{1}{2}
			\end{array}
			\right)
		\end{align}
	\end{enumerate}
\end{lemma}

\begin{lemma}\label{theorem:A2}
	The function $\mathbb{A}'_{\mathrm{Pauli}}(\rho)$ satisfies 
	\begin{align}
		\max\{\mathbb{A}'_{\mathrm{Pauli}}(\rho)\mid \rho\in\cS(\mathrm{SEP}^\ast(A;B))\}=3
	\end{align}
	and attains the maximum only on $\rho=\rho_{\mathrm{max}}$.
\end{lemma}

First, we show Lemma~\ref{theorem:A1}.

\begin{proof}[Proof of Lemma~\ref{theorem:A1}]
	\textbf{STEP1}:
	Proof of the statement 1.
	
	At first, the matrix given in \eqref{eq:state-A} is calculated as
	\begin{align}
		A:=
		\left(
		\begin{array}{cccc}
			1&0&0&0\\
			0&-1&2&0\\
			0&2&-1&0\\
			0&0&0&1
		\end{array}
		\right)
		=
		I-
		\left(
		\begin{array}{cccc}
			0&0&0&0\\
			0&2&-2&0\\
			0&-2&2&0\\
			0&0&0&0
		\end{array}
		\right).
	\end{align}
	the second matrix in right-hand-side is positive semi-definite with rank 1.
	Therefore, the maximum eigenvalue of $A$ is 1.,
	which implies that $\mathbb{A}'_{\mathrm{Pauli}}(\rho)\le 1$ for any positive semi-definite matrix $\rho$ with $\Tr \rho=1$. 
	
	\textbf{STEP2}:
	Proof of the statement 2.

	This is shown by the following simple calculation.
	\begin{align}
		\mathbb{A}'_{\mathrm{Pauli}}(\rho_{\mathrm{max}})=
		\Tr
		\left(
		\begin{array}{cccc}
			1&0&0&0\\
			0&-1&2&0\\
			0&2&-1&0\\
			0&0&0&1
		\end{array}
		\right)
		\left(
			\begin{array}{cccc}
				\frac{1}{2}&0&0&0\\
				0&0&\frac{1}{2}&0\\
				0&\frac{1}{2}&0&0\\
				0&0&0&\frac{1}{2}
			\end{array}
			\right)
		=3.
	\end{align}

	As the above,
	Lemma~\ref{theorem:A1} has been proven.
\end{proof}

Next, we show Lemma~\ref{theorem:A2},
In order to show Lemma~\ref{theorem:A2},
we apply the following proposition.
\begin{proposition}[{essentially shown in \cite[Prop 5.6]{CG2014} or \cite{Marciniak2010}}]\label{prop-ext}
	In $2\times2$-dimensional case,
	the set of all extremal points of $\cS(\mathrm{SEP}^\ast(A;B))$ is given as
	\begin{align}\label{eq-pure}
	\begin{aligned}
		&\mathrm{Ext}(\mathrm{SEP}^\ast(A;B))\\
		=&\{\Gamma(\rho)\mid\rho\in\cS(\mathrm{SES}(A;B)) \mbox{ $\rho$ is entangled pure}\}\\
		&\quad\quad\quad\quad\quad\cup\{\rho\mid\rho\in\cS(\mathrm{SES}(A;B)) \mbox{ $\rho$ is pure}\},
	\end{aligned}
	\end{align}
	where $\Gamma:=\id\otimes \top$ is the partial transposition map.
\end{proposition}
Actually, the reference \cite[Prop 5.6]{CG2014} shows that the all extremal points of the set of decomposable elements in $\cS(\mathrm{SEP}^\ast(A;B))$ is given as \eqref{eq-pure}.
It is known that all elements in $\mathrm{SEP}^\ast(A;B)$ are decomposable \cite{Peres1996} in $2\times2$-dimensional case,
and therefore, Proposition~\ref{prop-ext} holds.

\begin{proof}[Proof of Lemma~\ref{theorem:A2}]
	Because of Proposition~\ref{prop-ext},
	the maximum value of the function of $\mathbb{A}'_{\mathrm{Pauli}}(\rho)$ is attained by an element in $\mathrm{Ext}(A;B)$,
	i.e.,
	the maximum value is attained by a density matrix $\sigma$ with rank 1 or an element $\Gamma(\rho)$ of an density non-separable matrix $\rho$ with rank 1.
	Besides this,
	Lemma~\ref{theorem:A1}
	ensures that an element $\sigma$ with positive semi-definite rank 1 never attains the maximum value.
	Therefore,
	the maximum value of
	the function of $\mathbb{A}'_{\mathrm{Pauli}}(\rho)$ is attained by an element $\Gamma(\rho)$, where $\rho$ is positive semi-definite with rank 1.
	Then, the statement is obtained as follows:
	\begin{align}
		&\mathbb{A}'_{\mathrm{Pauli}}(\Gamma(\rho))+1
		\stackrel{(a)}{=}
		\Tr
		\left\{\left(
		\begin{array}{cccc}
			1&0&0&0\\
			0&-1&2&0\\
			0&2&-1&0\\
			0&0&0&1
		\end{array}
		\right)+I\right\}\Gamma(\rho)\nonumber\\
		=&\Tr \Gamma
		\left(
		\begin{array}{cccc}
			2&0&0&2\\
			0&0&0&0\\
			0&0&0&0\\
			2&0&0&2
		\end{array}
		\right)
		\Gamma(\rho)
		=\Tr 
		\left(
		\begin{array}{cccc}
			2&0&0&2\\
			0&0&0&0\\
			0&0&0&0\\
			2&0&0&2
		\end{array}
		\right)
		\rho\nonumber\\
		=&4\Tr\Phi_2\rho\stackrel{(b)}{\le} 4,
	\end{align}
	where $\Phi_2:=\frac{1}{2}\sum_{ij}e_{ij}\otimes e_{ij}$.
	The equation $(a)$ holds because $\Tr \Gamma(\rho)=1$.
	The inequality $(b)$ holds because $\Phi_2$ and $\rho$ are positive semi-definite projections.
\end{proof}

\begin{proof}[Proof of Theorem~\ref{theorem:A3}]
	The statement (2) is similarly shown by
	Lemma~\ref{theorem:A2}
	because any unitary matrix does not change the trace.
	We will show the statement (1).
	
	Let $\rho_0$ be a beyond-quantum pure state in $\cS(\mathrm{SEP}^\ast(A;B))$,
	i.e.,
	$\rho_0$ is written as $\Gamma(\rho)$, where $\rho$ is 
	an entangled state, i.e., a non-separable positive semi-definite matrix with trace 1 by Proposition~\ref{prop-ext}.
	Then, we obtain the following equation for any $U_A,U_B$:
	\begin{align}
		&\mathbb{A}_{\mathrm{Pauli}}(\Gamma(\rho);U_A,U_B)+1\nonumber\\
		=
		&\mathbb{A}'_{\mathrm{Pauli}}\left(\left(U_A\otimes U_B\right)\left(\Gamma(\rho)+I\right)\left(U_A^\dag\otimes U_B^\dag\right)\right)\nonumber\\
		=
		&\Tr
		\left(
		\begin{array}{cccc}
			2&0&0&0\\
			0&0&2&0\\
			0&2&0&0\\
			0&0&0&2
		\end{array}
		\right)\left(U_A\otimes U_B\right)\Gamma(\rho)\left(U_A^\dag\otimes U_B^\dag\right)\nonumber\\
		=
		&\Tr 4\Gamma(\Phi_2)\left(U_A\otimes U_B\right)\Gamma(\rho)\left(U_A^\dag\otimes U_B^\dag\right)\nonumber\\
		=&\Tr 4\Gamma(\Phi_2)\Gamma\left(U_A\otimes U_B'\right)\rho\left(U_A^\dag\otimes U_B^{\prime\dag}\right)\nonumber\\
		=&\Tr 4\Gamma(\Phi_2)\Gamma(\overline{\rho})=\Tr 4\Phi_2\overline{\rho},
	\end{align}
	where $U_B':=U_B^\top$ and $\overline{\rho}:=\left(U_A\otimes U_B'\right)\rho\left(U_A^\dag\otimes U_B^{\prime\dag}\right)$.
	Because $\overline{\rho}$ is entangled and $\Phi_2$ is maximally entangled in the $2\times 2$-dimensional bipartite quantum system,
	the inequality $\Tr \Phi_2\overline{\rho}>\frac{1}{2}$ holds,
	which implies the statement (2).

	As the above,
	Theorem~\ref{theorem:A3} has been proven.
\end{proof}

\end{document}